\documentclass[11pt]{article}

\makeatletter
\renewcommand*\@fnsymbol[1]{\the#1}
\makeatother

\usepackage{amsmath}
\usepackage{amsfonts}
\usepackage{amssymb}
\usepackage{amsthm}
\usepackage[english]{babel}
\usepackage{latexsym}
\usepackage[hmargin=2cm,vmargin=4cm]{geometry}
\usepackage{enumerate}
\usepackage{nicefrac}
\usepackage{mathrsfs}
\usepackage[affil-it]{authblk}

\theoremstyle{plain}
\newtheorem{theorem}{Theorem}[section]
\newtheorem{lemma}[theorem]{Lemma}
\newtheorem{proposition}[theorem]{Proposition}
\newtheorem{corollary}[theorem]{Corollary}
\theoremstyle{definition}
\newtheorem{definition}[theorem]{Definition}
\theoremstyle{remark}

\newtheorem{remark}[theorem]{Remark}
\newtheorem{example}[theorem]{Example}

\numberwithin{equation}{section}

\newcommand{\Interior}{\mathop{\rm int}\nolimits}
\newcommand{\e}{\varepsilon}
\newcommand{\VaR}{\mathop {\rm VaR}\nolimits}
\newcommand{\TVaR}{\mathop {\rm TVaR}\nolimits}

\newcommand{\cF}{{\mathcal F}}

\newcommand{\E}{{\mathbb E}}
\newcommand{\R}{{\mathbb R}}

\newcommand{\probp}{\mathbb P}

\newcommand{\cS}{{\mathscr{S}}}
\newcommand{\cA}{{\mathscr{A}}}
\newcommand{\cM}{{\mathscr{M}}}
\newcommand{\cX}{{\mathscr{X}}}
\newcommand{\cB}{{\mathscr{B}}}

\newcommand{\cU}{{\mathscr U}}

\def\keywords{\vspace{.5em}
{\noindent\textbf{Keywords}:\,\relax%
}}

\def\JELclassification{\vspace{.5em}
{\noindent\textbf{JEL classification}:\,\relax%
}}

\def\MSCclassification{\vspace{.5em}
{\noindent\textbf{MSC}:\,\relax%
}}

\makeatletter
\def\@fnsymbol#1{\ensuremath{\ifcase#1\or *\or 1\or 2\or
   3\or 4\or 5\or 6\or 7\or 8\else\@ctrerr\fi}}
\makeatother

\begin{document}

\title{Capital requirements with defaultable securities\footnote{Partial support through the SNF project 51NF40-144611 ``Capital adequacy, valuation, and portfolio selection for insurance companies'' is gratefully acknowledged.}}

\author{\sc{Walter Farkas}\thanks{Partial support by the National Centre of Competence in Research "Financial Valuation and Risk Management" (NCCR FinRisk), project "Mathematical Methods in Financial Risk Management", is gratefully acknowledged. Email: \texttt{walter.farkas@bf.uzh.ch}}\,, \sc{Pablo Koch-Medina}\thanks{Part of this research was undertaken when the author was employed by SwissRe. Email: \texttt{pablo.koch@bf.uzh.ch}}}
\affil{Department of Banking and Finance, University of Zurich, Switzerland}

\author{\sc{Cosimo Munari}\,\thanks{Financial support by the National Centre of Competence in Research "Financial Valuation and Risk Management" (NCCR FinRisk), project "Mathematical Methods in Financial Risk Management", is gratefully acknowledged. Email: \texttt{cosimo.munari@math.ethz.ch}}}
\affil{Department of Mathematics, ETH Zurich, Switzerland}

\date{November 22, 2013}

\maketitle

\begin{abstract}
We study capital requirements for bounded financial positions defined as the minimum amount of capital to invest in a chosen {\em eligible} asset targeting a pre-specified \textit{acceptability} test. We allow for general acceptance sets and general eligible assets, including defaultable bonds. Since the payoff of these assets is not necessarily bounded away from zero the resulting risk measures cannot be transformed into cash-additive risk measures by a change of numeraire. However, extending the range of eligible assets is important because, as exemplified by the recent financial crisis, assuming the existence of default-free bonds may be unrealistic. We focus on finiteness and continuity properties of these general risk measures. As an application, we discuss capital requirements based on Value-at-Risk and Tail-Value-at-Risk acceptability, the two most important acceptability criteria in practice. Finally, we prove that there is no optimal choice of the eligible asset. Our results and our examples show that a theory of capital requirements allowing for general eligible assets is richer than the standard theory of cash-additive risk measures.
\end{abstract}


\MSCclassification{91B30, 91B32}

\JELclassification{C60, G11, G22}

\parindent 0em \noindent

\section{Introduction}

The objective of this paper is to investigate capital requirements for bounded financial positions in a world where default-free securities do not necessarily exist. As we will see, this general problem cannot be treated within the standard theory of cash-additive risk measures. Hence, our work extends and complements the literature on cash-additive risk measures on spaces of bounded measurable functions. By doing so we hope to also contribute to a more informed application of the theory of risk measures to capital adequacy issues arising in the design of modern solvency regimes.

\medskip

Liability holders of a financial institution are credit sensitive: they, and regulators on their behalf, are concerned that the institution may fail to fully honor its future obligations. This will be the case if the institution's {\em financial position}, or \textit{capital position} -- the value of assets net of liabilities -- becomes negative in some future state of the economy. To address this concern financial institutions hold {\em risk capital} whose function is to absorb unexpected losses thereby reducing the likelihood that they may become insolvent. A key question in this respect is how much capital a financial institution should be required to hold to be deemed adequately capitalized by the regulator.

\medskip

This type of question is best framed using the concepts of an \textit{acceptance set} and of a {\em risk measure}. Coherent acceptance sets and coherent risk measures were introduced in the seminal paper by Artzner et al.\,(1999) for finite sample spaces and by Delbaen (2002) for general probability spaces. Convex risk measures were studied by F\"{o}llmer and Schied (2002) and by Frittelli and Rosazza Gianin (2002). Since then the theory of risk measures has established itself as the standard theoretical framework to approach the problem of capital adequacy in financial institutions, and continues to influence the debate on modern solvency regimes in both the insurance and the banking world.

\medskip

In this paper we work within a one-period model with dates~$t=0$ and~$t=T$, and assume financial positions at time~$T$ belong to the space~$\cX$ of bounded measurable functions $X:\Omega\to\R$ on a given measurable space $(\Omega,\cF)$. At the core of the theory is the concept of an acceptance set~$\cA\subset\cX$, representing the set of future capital positions corresponding to financial institutions that are deemed to be well capitalized. Once an acceptability criterion has been defined it is natural to ask whether the management of a badly capitalized financial institution can achieve acceptability by implementing appropriate actions and, if so, at what cost. The theory of risk measures was designed to answer this question for a particular type of management action: raising capital and investing in a reference traded asset, the so-called \textit{eligible} asset.

\medskip

In this framework, required capital is defined as the minimum amount of capital that, when invested in the eligible asset, makes a given financial position acceptable. Formally, if $\cA\subset\cX$ is the chosen acceptance set and $S=(S_0,S_T)$ represents a traded asset with initial value $S_0>0$ and positive payoff $S_T\in\cX$, the risk capital required for a position $X\in\cX$ is given by
\begin{equation}
\label{risk measure intro}
\rho_{\cA,S}(X):=\inf\left\{m\in\R \,; \ X+\frac{m}{S_0}S_T\in\cA\right\}\,.
\end{equation}

\medskip

The bulk of the literature on capital requirements has focused on cash-additive risk measures
\begin{equation}
\rho_{\cA}(X):=\rho_{\cA,B}(X)=\inf\left\{m\in\R \,; \ X+m\in\cA\right\}\,,
\end{equation}
for which the eligible asset $B=(B_0,B_T)$ is a risk-free bond with $B_0=1$ and $B_T=1_\Omega$. This case is in fact more general than it may seem at first. Indeed, consider the situation where the payoff~$S_T$ of the eligible asset is bounded away from zero and choose~$S$ as the numeraire. With respect to this new numeraire, any financial position $X\in\cX$ and the acceptance set~$\cA$ take their ``discounted'' form $\widetilde{X}:=X/S_T$ and $\widetilde{\cA}:=\left\{X/S_T \,; \ X\in\cA\right\}$, respectively, and
\begin{equation}
\rho_{\cA,S}(X)=S_0\,\rho_{\widetilde{\cA}}\,(\widetilde{X})\,.
\end{equation}
The risk measure~$\rho_{\cA,S}$ can be therefore expressed in terms of the cash-additive risk measure~$\rho_{\widetilde{\cA}}$.

\medskip

Our starting point is the following observation: \textit{the artifice of changing the numeraire does not work for all choices of the eligible asset}. Indeed, assume the eligible asset $S=(S_0,S_T)$ is a defaultable bond with price $S_0\in(0,1)$ and face value~$1_\Omega$. Since the bond is defaultable, in some future state $\omega\in\Omega$ the payoff $S_T(\omega)$ may be strictly smaller than the face value. Hence,~$S_T$ is a random variable taking values in the interval~$[0,1]$ and represents the ``recovery rate'', i.e. the portion of the face value that is recovered by the bond holder. If~$S_T$ is not bounded away from zero, which includes the possibility that it is zero in some future state, the risk measure~$\rho_{\cA,S}$ \textit{cannot be reduced to a cash-additive risk measure} by changing the numeraire. If the recovery rate is always strictly positive but not bounded away from zero, the discounted positions $\widetilde{X}:=X/S_T$ make sense but no longer belong to~$\cX$ in general. Hence, we can no longer operate in the space of bounded measurable functions and the resulting underlying space will depend on the choice of the numeraire. If, on the other hand,  the recovery rate is zero in some future scenario, then it is not even meaningful to speak about discounted positions. As a consequence, focusing on cash-additive risk measures only, rules out the possibility that the eligible asset is a defaultable bond with a recovery rate which is not bounded away from zero.

\medskip

As the recent financial crisis has made painfully clear, assuming the existence of default-free bonds may turn out to be delusive. Indeed, even recovery rates which are zero in some future scenario are not unrealistic. Zero recovery rates  arise naturally also in situations where actual recovery is strictly positive but may come too late to be of practical relevance in the capital adequacy assessment. In such cases, for solvency purposes, it may be necessary to assume zero recovery. Hence, to obtain a more realistic theory of capital requirements allowing for the possibility that the eligible asset is a defaultable bond, \textit{we are forced to go beyond cash-additivity and consider risk measures~$\rho_{\cA,S}$ with respect to general eligible assets}.

\medskip

Given an acceptance set $\cA$ and a general eligible asset $S=(S_0,S_T)$ we will focus on the following three questions:
\begin{enumerate}
\item When is $\rho_{\cA,S}$ finitely valued? This is an important question, also from an economic perspective. Indeed, if $\rho_{\cA,S}(X)=-\infty$ for a position $X\in\cX$, then we can extract arbitrary amounts of capital from the financial institution while retaining acceptability, which is clearly not economically meaningful. If, on the other hand, $\rho_{\cA,S}(X)=\infty$, then~$X$ cannot be made acceptable no matter how much capital we raise and invest in the eligible asset. Hence, the choice of the eligible asset is not ``effective'' when it comes to modifying the acceptability of~$X$.

\item When is $\rho_{\cA,S}$ continuous? This is also of practical relevance since typically capital positions are based on estimates and can only be assessed in an approximate manner. Thus, it is important to know whether capital requirements are ``stable'' with respect to small perturbations of the capital position.

\item Is it possible to find an optimal eligible asset leading to the lowest risk capital compatible with a given acceptability criterion? This is related to the ``efficiency'' of the choice of the eligible asset, i.e. to the ability to reach acceptance with the least possible amount of capital.
\end{enumerate}

\medskip

In addressing the first two questions we will find that if we allow for general eligible assets, the range of possible behaviors is much broader than in the standard cash-additive setting, where every risk measure on $\cX$ is finitely valued and globally Lipschitz continuous. We will exhibit examples of capital requirements which are neither finitely valued nor continuous. This is even the case when the underlying acceptability criterion is based on Value-at-Risk or Tail Value-at-Risk, which are the two typical choices in modern regulatory environments.

\medskip

As a consequence of general results on finiteness (Theorem~\ref{conic finiteness}) and continuity (Theorem~\ref{convexity and continuity} and Theorem~\ref{pointwise semicontinuity}), we will show that capital requirements based on Value-at-Risk are not always finitely valued (Proposition~\ref{ruin and eligibility}), and, even when finitely valued, are not always continuous (Proposition~\ref{continuity var based}). Also capital requirements based on Tail Value-at-Risk need not be finitely valued (Proposition~\ref{TVaR-acceptance}), but, whenever finite, they are also continuous (Proposition~\ref{continuity of TVaR}). In fact, for capital requirements based on Value-at-Risk and Tail-Value-at-Risk acceptability we will provide complete characterizations of finiteness and continuity. In the case that the eligible asset is a defaultable bond, these characterizations show that finiteness and continuity depend on the extent to which the issuer of the bond can default. In particular, our results show that, when the underlying probability space is nonatomic, capital requirements based on Value-at-Risk are continuous if and only if the recovery rate is bounded away from zero. This is in contrast to Tail-Value-at-Risk acceptability, whose corresponding capital requirements may be continuous even in cases where the recovery rate is zero in some future scenario.

\medskip

With respect to the third question we will show (Theorem~\ref{general optimality theorem}) that no optimal choice of the eligible asset is possible. This extends a result in Artzner et al.\,(2009) obtained in the case of coherent acceptance sets in finite sample spaces to the general case where neither coherence nor convexity requirements are imposed and where we allow for infinite sample spaces. Consequently, our result also applies to risk measures based on Value-at-Risk -- which, as is well known, fail to be convex in general.

\medskip

In this paper we consider financial positions in spaces of bounded measurable functions, the original setting in which capital requirements were studied. In the recent literature other types of spaces such as $L^p$-spaces, Orlicz spaces, or even abstract ordered spaces have been considered -- see for instance Kaina and R\"{u}schendorf~(2009), Cheridito and Li~(2009), Biagini and Frittelli~(2009) and Konstantinides and Kountzakis~(2011) and the literature cited therein. Those papers deal mainly with dual representation theorems and continuity properties of risk measures. Moreover, they focus on risk measures that are either cash-additive or additive with respect to eligible assets which are interior points of the positive cone, which in our case are precisely those eligible assets whose payoff is bounded away from zero. Therefore, none of those results can be applied to the type of situation we are interested in. At the same time, the results here exploit the particular structure of spaces of bounded measurable functions -- most importantly the fact that their positive cones have nonempty interior, which amongst other things implies that all acceptance sets have nonempty interior. Hence, the techniques used here do not transfer easily to~$L^p$ or Orlicz spaces, whose positive cones generally have empty interior. For a treatment of risk measures in this kind of situation we refer to Farkas et. al (2014).


\section{Preliminaries}

We consider a single-period economy with dates~$t=0$ and~$t=T$ and a measurable space $(\Omega,\cF)$ representing uncertainty at time~$T$. Thus, the elements of~$\Omega$ are interpreted as the possible future states of the economy.


\subsection{Financial positions and acceptance sets}
\label{financial positions}

We assume that the capital position of a financial institution at time~$T$ belongs to the Banach space~$\cX$ of real-valued, bounded, $\cF$-measurable functions on~$\Omega$ equipped with the supremum norm
\begin{equation}
\left\|X\right\|:=\sup_{\omega\in\Omega}\left|X(\omega)\right|\,.
\end{equation}
If $A\subset\Omega$ we denote by~$1_A$ the associated indicator function. For a subset $\cA\subset\cX$ we denote by~$\Interior(\cA)$, \,$\overline{\cA}$ and~$\partial\cA$ the interior, the closure and the boundary of~$\cA$, respectively.

\medskip

The space~$\cX$ becomes a Banach lattice when endowed with the pointwise ordering defined by $Y\geq X$ whenever $Y(\omega)\geq X(\omega)$ for all $\omega\in\Omega$. The corresponding positive cone is then given by $\cX_+:=\{X\in\cX \,;\, X\ge 0\}$. Clearly,~$\Interior(\cX_+)$ consists of those positions~$X$ in~$\cX_+$ that are \textit{bounded away from zero}, i.e. such that for some $\e>0$ we have $X\geq\e1_{\Omega}$.

\medskip

We now recall the concept of an acceptance set, which is central to this paper.

\begin{definition}
A subset $\cA\subset\cX$ is called an {\em acceptance set} if it satisfies the following two axioms:
\begin{enumerate}
\item [(A1)] $\cA$ is a nonempty, proper subset of~$\cX$ (non-triviality);\label{non-triviality axiom}
\item [(A2)] if $X\in\cA$ and $Y\geq X$ then $Y\in\cA$ (monotonicity).\label{monotonicity axiom}
\end{enumerate}
\end{definition}

\medskip

An acceptance set represents the set of future positions that are deemed to provide a reasonable security to the liability holders. Hence, testing whether a financial institution is adequately capitalized or not with respect to a chosen acceptance set~$\cA$ reduces to establishing whether its financial position belongs to~$\cA$ or not. Therefore, one could refer to~$\cA$ as being a {\em capital adequacy test}.

\medskip

\begin{remark}
\begin{enumerate}[(i)]
\item The defining properties of an acceptance set encapsulate what one would expect from any non-trivial capital adequacy test: some -- but not all -- positions should be acceptable, and
any financial position dominating an already accepted position should also be acceptable.
\item \textit{Coherent} acceptance sets, introduced by Artzner et al.\,(1999), are acceptance sets that are convex cones. \textit{Convex} acceptance sets were introduced by F\"{o}llmer and Schied~(2002) and by Frittelli and Rosazza Gianin~(2002). Note that in the above definition of an acceptance set neither coherence nor convexity is required. We will also pay particular attention to \textit{conic} acceptance sets, i.e. acceptance sets which are cones, of which the acceptance set based on Value-at-Risk, introduced in Example~\ref{var tvar example}, is a prominent representative.
\end{enumerate}
\end{remark}

\medskip

The following result summarizes some useful properties of acceptance sets in the context of spaces of bounded financial positions.

\begin{lemma}
\label{properties acceptance sets}
Let $\cA\subset\cX$ be an arbitrary acceptance set. Then the following statements hold:
\begin{enumerate}[(i)]
	\item $\cA$ contains all sufficiently large constants;
	\item $X\in\Interior(\cA)$ if and only if $X-\e 1_{\Omega}\in \cA$ for some $\e>0$;
	\item $X\notin\overline{\cA}$ if and only if $X+\e 1_{\Omega}\not\in\cA$ for some $\e>0$;
	\item $\Interior(\cA)$ is an acceptance set, in particular $\Interior(\cA)\neq\emptyset$, and $\overline{\Interior(\cA)}=\overline{\cA}$;
	\item $\overline{\cA}$ is an acceptance set, in particular $\overline{\cA}\neq\cX$, and $\Interior(\overline{\cA})=\Interior(\cA)$.
\end{enumerate}
\end{lemma}
\proof
Since $\cA\neq\emptyset$, assertion~{\em (i)} is clear by monotonicity. Statements~{\em (ii)} and~{\em (iii)} also follow easily using the monotonicity of~$\cA$.

\smallskip

To prove {\em (iv)} note that $\Interior(\cA)\neq\emptyset$ since~$\cX_+$ has nonempty interior and, by monotonicity,~$\cA$ contains a translate of~$\cX_+$. Take $X\in\Interior(\cA)$ so that $X+\cU\subset\cA$ for some neighborhood of zero~$\cU$. As a result, for any $Y\geq X$ we have $Y+\cU\subset\cA$, implying that~$\Interior(\cA)$ is monotone. To conclude the proof of~\textit{(iv)} it is enough to note that $\cA\subset\overline{\Interior(\cA)}$. Indeed, if $X\in\cA$, then $X_n:= X+\frac{1}{n}1_\Omega$ defines a sequence in~$\Interior(\cA)$ converging to~$X$.

\smallskip

To prove~{\em (v)} we first establish that $\Interior(\overline{\cA})=\Interior(\cA)$ by showing $\Interior(\overline{\cA})\subset\cA$. Indeed, take $X\in\Interior(\overline{\cA})$. Then $X-\e 1_\Omega\in\overline{\cA}$ for a suitably small $\e>0$. Hence, for $0<\delta<\e$ we find $Y\in\cA$ with $\left\|X-\e 1_\Omega-Y\right\|<\delta$, implying $Y\leq X-\e 1_\Omega+\delta 1_\Omega\leq X$. Then, by monotonicity, $X\in\cA$ and therefore the interior of~$\overline{\cA}$ coincides with~$\Interior(\cA)$. In particular, $\overline{\cA}\neq \cX$. To prove monotonicity take $X\in\overline{\cA}$ and $Y\geq X$. Let~$(X_n)$ be a sequence in~$\cA$ converging to~$X$. Since $Y_n:=X_n+Y-X$ defines a sequence with limit~$Y$ and such that $Y_n\geq X_n$, we conclude that~$\overline{\cA}$ is monotone.
\endproof

\medskip

A detailed treatment of Value-at-Risk and Tail-Value-at-Risk, the basis for the following example, can be found in Section~4.4 in F\"{o}llmer and Schied~(2011).

\smallskip

\begin{example}[$\VaR$- and $\TVaR$-acceptability]
\label{var tvar example}
Assume a probability~$\probp$ is given on~$(\Omega,\cF)$ and take~$\alpha\in(0,1)$.

\begin{enumerate}[(i)]
\item The {\em Value-at-Risk of $X\in\cX$ at the level $\alpha$} is defined as
\begin{equation}
\VaR_\alpha(X):=\inf\{m\in\R \,; \ \probp[X+m<0]\le\alpha\}\;.
\end{equation}
The set
\begin{equation}
\label{var acceptance set}
\cA_\alpha:=\{X\in\cX \,; \ \VaR_\alpha(X)\le 0\}=\{X\in\cX \,; \ \probp[X<0]\le\alpha \}
\end{equation}
is an acceptance set which is a closed cone, but which fails, in general, to be convex. The second description of~$\cA_\alpha$ in~\eqref{var acceptance set} shows that $\VaR$-acceptability at level~$\alpha$ is equivalent to requiring that the probability of ruin of an institution is not exceeding the threshold~$\alpha$ and helps explain its popularity. Finally, we note for later use that $\Interior(\cA_\alpha)=\{X\in\cX \,;\, \VaR_\alpha(X)< 0\}$. This follows from the well-known continuity of the cash-additive risk measure $\VaR_\alpha:\cX\to\R$.

\item The {\em  Tail Value-at-Risk of $X\in\cX$ at the level $\alpha$} is defined as
\begin{equation}
\TVaR_\alpha(X):=\frac{1}{\alpha} \int_0^\alpha\VaR_\beta(X)d\beta\;.
\end{equation}
Tail Value-at-Risk is also known under the names of {\em Expected Shortfall},
{\em Conditional Value-at-Risk}, or {\em Average Value-at-Risk}. The set
\begin{equation}
\cA^\alpha:=\{X\in\cX \,;\, \TVaR_\alpha(X)\le 0\}
\end{equation}
is probably the most well-known example of a closed, coherent acceptance set. Note also here that $\Interior(\cA^\alpha)=\{X\in\cX \,;\, \TVaR_\alpha(X)<0\}$, as implied by the continuity of the cash-additive risk measure $\TVaR_\alpha:\cX\to\R$.
\end{enumerate}
\end{example}


\subsection{Required capital and risk measures}
\label{Risk measures with respect to a single eligible asset}

In order to introduce the notion of a capital requirement, we consider traded assets~$S=(S_0,S_T)$ with initial value $S_0>0$ and nonzero payoff~$S_T\in\cX_+$. The following definition goes back to Definition~2.2 in Artzner et al.\,(1999). As usual, we denote by $\overline{\R}:=\R\cup\{-\infty, \infty\}$ the extended real line.

\medskip

\begin{definition}
\label{risk measure single elig def}
Let~$\cA$ be an arbitrary subset of~$\cX$ and $S=(S_0,S_T)$ a traded asset. The {\em capital requirement} or \textit{risk measure} with respect to~$\cA$ and~$S$ is the map $\rho_{\cA,S}:\cX\to\overline{\R}$ defined by setting for $X\in\cX$
\begin{equation}
\label{risk measure formula}
\rho_{\cA,S}(X):=\inf\left\{m\in\R \,; \ X+\frac{m}{S_0}S_T\in\cA\right\}\,.
\end{equation}

\smallskip

The asset~$S$ is called the {\em reference} or {\em eligible} asset.
\end{definition}

\medskip

The quantity~$\rho_{\cA,S}(X)$ represents a capital amount. If~$\cA$ is an acceptance set and $\rho_{\cA,S}(X)$ is finite and positive, then it intuitively represents the minimum investment in the asset~$S$ that is required to make the unacceptable financial position~$X$ acceptable or, put differently, the cost of making~$X$ acceptable. Similarly, if finite and negative, we interpret $\rho_{\cA,S}(X)$ as the amount of capital that can be extracted from a well-capitalized financial institution without compromising the acceptability of its financial position.

\smallskip

\begin{remark}
\label{making acceptable remark}
\begin{enumerate}[(i)]
  \item Let~$\cA$ be an arbitrary subset of $\cX$. If $\rho_{\cA,S}(X)\in\R$ for a position $X\in\cX$, then we can only infer that $X+\frac{\rho_{\cA,S}(X)}{S_0}S_T\in\overline{\cA}$ but not necessarily $X+\frac{\rho_{\cA,S}(X)}{S_0}S_T\in\cA$. Hence, it is only in this approximate sense that $\rho_{\cA,S}(X)$ represents the ``minimum'' amount of capital required to make a position acceptable. Clearly, if~$\cA$ is closed the the infimum in~\eqref{risk measure formula} is always attained.
  \item Note that if $S=(S_0,S_T)$ is the risk-free asset with $S_0=1$ and $S_T=1_\Omega$, the previous definition coincides with the standard definition of a (cash-additive) risk measure as given for example in Definition~4.1 in F\"{o}llmer and Schied~(2011). In this case we will simply write~$\rho_\cA$.
\end{enumerate}
\end{remark}

\smallskip

\begin{remark}
Note that the eligible asset can be taken to be any traded asset with nonzero, positive payoff, even if the payoff is zero in some future state. In particular, defaultable bonds, options, and limited-liability assets all qualify as eligible assets.
\end{remark}

\medskip

\begin{definition}
\label{definition of abstract risk measures}
Let $S=(S_0,S_T)$ be a traded asset. A map $\rho:\cX\to\overline{\R}$ is called \textit{$S$-additive} if
\begin{equation}
\rho(X+\lambda S_T)=\rho(X)-\lambda S_0 \ \ \mbox{for all} \ X\in\cX \ \mbox{and} \ \lambda\in\R\,,
\end{equation}
and is called \textit{decreasing} if
\begin{equation}
\rho(X)\geq\rho(Y) \ \ \mbox{whenever} \ X\leq Y\,.
\end{equation}
\end{definition}

\medskip

The proposition below collects the main properties of~$\rho_{\cA,S}$, which are the general counterparts to the well-known properties of cash-additive risk measures as stated in Section~4.1 in F\"{o}llmer and Schied~(2011). The proof is straightforward and is omitted.

\begin{lemma}
\label{properties proposition}
Let~$\cA$ be an arbitrary subset of~$\cX$ and $S=(S_0,S_T)$ a traded asset. The following statements hold:
\begin{enumerate}[(i)]
\item  $\rho_{\cA,S}$ is $S$-additive;
\item if $\cA$ satisfies the monotonicity axiom {\rm (A2)}, then $\rho_{\cA,S}$ is decreasing;
\item if $\cA$ satisfies the non-triviality axiom {\rm (A1)}, then $\rho_{\cA,S}$ cannot be identically~$\infty$ or~$-\infty$;
\item if $\rho:\cX\to\overline{\R}$ is $S$-additive, decreasing, and not identically equal to~$\infty$ or~$-\infty$, then the set $\cA_\rho:=\{X\in\cX \,; \ \rho(X)\leq 0\}$ is an acceptance set and $\rho=\rho_{\cA_\rho,S}$.
\end{enumerate}
\end{lemma}

\medskip

The following result is particularly useful when studying continuity properties of capital requirements.

\begin{lemma}
\label{inclusions A and rho}
Let $\cA\subset\cX$ be an acceptance set and~$S=(S_0,S_T)$ a traded asset. Then we have
\begin{equation}
\label{chain of inclusions}
\Interior(\cA)\subset\{X\in\cX \,; \ \rho_{\cA,S}(X)<0\}\subset\cA\subset\{X\in\cX \,; \ \rho_{\cA,S}(X)\leq0\}\subset\overline{\cA}\,.
\end{equation}
Moreover, we have $\{X\in\cX \,; \ \rho_{\cA,S}(X)=0\}\subset\partial\cA$.
\end{lemma}
\begin{proof}
First, note that $\rho_{\cA,S}(X)\leq0$ whenever $X\in\cA$ and $\rho_{\cA,S}(X)\geq0$ whenever $X\notin\cA$, proving the two central inclusions in~\eqref{chain of inclusions}.

\smallskip

To prove the first inclusion in~\eqref{chain of inclusions}, take $X\in\Interior(\cA)$. Then $X-\lambda S_T\in\cA$ for some $\lambda>0$. As a result, $\rho_{\cA,S}(X)+\lambda S_0=\rho_{\cA,S}(X-\lambda S_T)\leq0$, implying that $\rho_{\cA,S}(X)<0$.

\smallskip

To prove the last inclusion in~\eqref{chain of inclusions}, assume $X\notin\overline{\cA}$. Since $X+\lambda S_T\notin\cA$ for a suitably small $\lambda>0$, we obtain $\rho_{\cA,S}(X)-\lambda S_0=\rho_{\cA,S}(X+\lambda S_T)\geq0$. It follows that $\rho_{\cA,S}(X)>0$, concluding the proof. As a consequence of~\eqref{chain of inclusions} we also obtain that $X\in\partial\cA$ if $\rho_{\cA,S}(X)=0$.
\end{proof}

\medskip

\begin{remark}[Change of numeraire]
\label{discounted positions note}
Assume $\cA\subset\cX$ is an acceptance set and $S=(S_0,S_T)$ a traded asset such that~$S_T$ is bounded away from zero. Write $\widetilde X:=X/S_T$ and note that $\widetilde X\in\cX$ for any $X\in\cX$. If we introduce the acceptance set $\widetilde\cA:=\{\widetilde{X}\in\cX \,; \ X\in\cA\}$, we have for every $X\in\cX$
\begin{equation}
\label{change of numeraire}
\rho_{\cA,S}(X)=S_0\,\rho_{\widetilde\cA}\,(\widetilde X)\,.
\end{equation}
Hence by a simple change of numeraire we can transform the risk measure $\rho_{\cA,S}$ into the cash-additive risk measure $\rho_{\widetilde\cA}$. However, as argued in the introduction, this device does not work if~$S_T$ is not bounded away from zero. For this reason, focusing on cash-additive risk measures based on the argument that they correspond to risk measures for ``discounted'' positions implicitly assumes that the payoff of the actual eligible asset is bounded away from zero. This assumption critically limits the choice of such asset and inhibits the use of important classes including defaultable bonds.

Because of the above mentioned ``discounting'' argument, the eligible asset is sometimes referred to as the {\em numeraire asset}. We prefer to use the term {\em eligible asset} since it emphasizes that the asset has been chosen as the vehicle to define required capital. The fact that, if at all possible, it may sometimes be convenient to choose the eligible asset as the unit of account is merely incidental.
\end{remark}

\medskip

It is well known that cash-additive risk measures in the space of bounded measurable functions are always finitely valued and Lipschitz continuous, see for instance Lemma~4.3 in F\"{o}llmer and Schied~(2011). Hence the following result is an immediate consequence of the representation~\eqref{change of numeraire}.

\begin{proposition}
\label{finiteness baz}
Let $\cA\subset\cX$ be an arbitrary acceptance set and assume $S=(S_0,S_T)$ is a traded asset with payoff~$S_T$ bounded away from zero. Then $\rho_{\cA,S}$ is finitely valued and Lipschitz continuous.
\end{proposition}

\smallskip

\begin{remark}
\label{risk-free asset}
Note that in case~$\Omega$ is finite,~$S_T$ is bounded away from zero whenever $S_T(\omega)>0$ for all $\omega\in\Omega$. It follows that for finite~$\Omega$ the risk measure~$\rho_{\cA,S}$ is finitely valued for any acceptance set~$\cA$ and any eligible asset~$S$ with everywhere positive payoff. This is the setting in Artzner et al.\,(1999).
\end{remark}

\smallskip

As we will see throughout this paper, risk measures with respect to a general eligible asset display a much wider variety of behaviors in terms of finiteness and continuity than cash-additive risk measures.


\section{Finiteness of risk measures}
\label{finiteness}

Given an acceptance set~$\cA$ and an eligible asset $S=(S_0,S_T)$ it is natural to ask when $\rho_{\cA,S}$ is finitely valued, i.e. when $-\infty<\rho_{\cA,S}(X)<\infty$ holds for every $X\in\cX$. This is not only an interesting mathematical question but has also an economic significance. Indeed, if a position~$X$ is such that $\rho_{\cA,S}(X)=\infty$, then it cannot be made acceptable by raising any amount of capital and investing it in the eligible asset. This suggests that the eligible asset is not ``effective'' for this position. On the other hand, if $\rho_{\cA,S}(X)=-\infty$, then we can borrow arbitrary amounts of capital by shorting the eligible asset while retaining the acceptability of $X$. This might be called {\em acceptability arbitrage} and is related to the same concept introduced in the context of multiple eligible assets in Artzner et al.\,(2009) (see in particular Assumption $NAA(\cA,\cS)$ and Remark~1 following Proposition~2 in that paper).


\subsection{The geometry of finiteness}

The following basic result is useful when establishing whether a capital requirement is finitely valued or not. We will use it without further reference. The proof is a direct consequence of the definition of a capital requirement.

\begin{proposition}
\label{finiteness lemma}
Let $\cA\subset\cX$ be an acceptance set, and $S=(S_0,S_T)$ a traded asset. For $X\in\cX$ the following statements hold:
\begin{enumerate}[(i)]
	\item $\rho_{\cA,S}(X)<\infty$ if and only if there exists $\lambda_0\in\R\cup\{-\infty\}$ with $X+\lambda S_T\in\cA$ for any $\lambda>\lambda_0$;
	\item $\rho_{\cA,S}(X)>-\infty$ if and only if there exists $\lambda_0\in\R\cup\{\infty\}$ with $X+\lambda S_T\notin\cA$ for any~$\lambda<\lambda_0$.
\end{enumerate}
\end{proposition}

\medskip

The set of examples below shows that if we depart from the usual cash-additive framework and we allow for general eligible assets, capital requirements are no longer automatically finite. In particular, it is not true that any financial position can be made acceptable by adding to it a suitable amount of capital invested in the eligible asset.

\smallskip

\begin{example}
\label{non reachability example}
In all the following examples $S=(S_0,S_T)$ denotes the eligible asset.
\begin{enumerate}[(i)]
  \item {\em A risk measure with range $\R\cup\{\infty\}$.} Set $\cA:=\cX_+$. Then, if~$S_T$ is not bounded away from zero, it is easy to see that $\rho_{\cA,S}$ cannot assume the value~$-\infty$. However, $\rho_{\cA,S}(-1_\Omega)=\infty$ and $\rho_{\cA,S}(S_T)=-S_0$.

Note that this implies that $\rho_{\cA,S}$ is finitely valued if and only if the payoff~$S_T$ is bounded away from zero.

  \item {\em A risk measure with range $\R\cup\{-\infty\}$.} Let~$A$ and~$B$ be disjoint nonempty measurable subsets of~$\Omega$ such that $A\cup B=\Omega$. Set $\cA:=\{X\in\cX \,; \ X1_A \ge 0 \ \text{or} \ X1_B>0\}$ and $S_T:=1_B$. Then $\rho_{\cA,S}$ cannot assume the value~$\infty$. However, $\rho_{\cA,S}(1_A)=-\infty$ and $\rho_{\cA,S}(-1_\Omega)=S_0$.

  \item {\em A risk measure with range $\overline{\R}$.} Let~$\cA$ be as in the previous example but assume~$C$ is a nonempty measurable proper subset of~$B$. Set $S_T:=1_C$. Then we have $\rho_{\cA,S}(1_A)=-\infty$, $\rho_{\cA,S}(-1_\Omega)=\infty$, and $\rho_{\cA,S}(-1_A-1_C)=S_0$.

  \item {\em A risk measure with range $\{-\infty,\infty\}$ (even though~$S_T$ is everywhere strictly positive).} Assume $Z\in\cX$ is everywhere strictly positive but not bounded away from zero. Consider the acceptance set $\cA:=\{X\in\cX \,; \ \exists \ \lambda\in\R \,:\, X\geq\lambda Z\}$. Assume $S_T\leq\lambda_0 Z$ for some $\lambda_0>0$. Then, $\rho_{\cA,S}(X)=-\infty$ whenever $X\in\cA$ and $\rho_{\cA,S}(X)=\infty$ whenever $X\not\in\cA$.
\end{enumerate}
\end{example}


\subsection{Conic and coherent capital requirements}

When the acceptance set is a cone we can characterize those eligible assets for which the resulting capital requirements are finitely valued.

\begin{theorem}
\label{conic finiteness}
Let $\cA\subset\cX$ be a conic acceptance set and $S=(S_0,S_T)$ a traded asset. The following statements hold:
\begin{enumerate}[(i)]
	\item $\rho_{\cA,S}(X)<\infty$ for all $X\in\cX$ if and only if $S_T\in\Interior(\cA)$;
	\item $\rho_{\cA,S}(X)>-\infty$ for all $X\in\cX$ if and only if $-S_T\notin\overline{\cA}$.
\end{enumerate}
\end{theorem}
\begin{proof}
\textit{(i)} Assume that $S_T\in\Interior(\cA)$, and take an arbitrary $X\in\cX$. Then for small $\lambda>0$ we have $S_T+\lambda X\in\cA$. Since~$\cA$ is a cone we immediately obtain $X+\frac{1}{\lambda}S_T\in\cA$, implying $\rho_{\cA,S}(X)<\infty$. Assume now~$\rho_{\cA,S}$ does not take the value~$\infty$. Then, in particular, there exists $\lambda>0$ such that $-1_\Omega +\lambda S_T\in\cA$. Since~$\cA$ is a cone we get $S_T-\frac{1}{\lambda}1_\Omega \in\cA$, which by Lemma~\ref{properties acceptance sets} yields $S_T\in\Interior(\cA)$.

\smallskip

\textit{(ii)} Assume that $-S_T\notin\overline{\cA}$ and let $X\in\cX$. Then there exists $\lambda>0$ such that $-S_T+\lambda X\notin\cA$. As a result $X-\frac{1}{\lambda}S_T\notin\cA$, showing that $\rho_{\cA,S}(X)>-\infty$. Conversely, assume $\rho_{\cA,S}$ is never~$-\infty$. Then we can find $\lambda>0$ such that $1_\Omega-\lambda S_T\notin\cA$. Since~$\cA$ is a cone, this is equivalent to $-S_T+\frac{1}{\lambda}1_\Omega\notin\cA$. Therefore $-S_T\notin\overline{\cA}$ by Lemma~\ref{properties acceptance sets}, concluding the proof.
\end{proof}

\medskip

If the acceptance set is coherent we obtain a stronger characterization.

\begin{corollary}
\label{coherent finiteness}
Let $\cA\subset\cX$ be a coherent acceptance set and $S=(S_0,S_T)$ a traded asset. Then the following statements are equivalent:
\begin{enumerate}[(a)]
  \item $\rho_{\cA,S}$ is finitely valued;
  \item $\rho_{\cA,S}$ does not take the value $\infty$;
  \item $S_T\in\Interior(\cA)$.
\end{enumerate}
\end{corollary}
\begin{proof}
Assume $S_T\in\Interior(\cA)$. Since~$\cA$ is convex, Lemma~5.28 in Aliprantis and Border~(2006) implies that if $-S_T\in\overline{\cA}$ we have $-\lambda S_T+(1-\lambda)S_T\in\Interior(\cA)$ for every $\lambda\in(0,1)$. In particular, $0\in\Interior(\cA)$. Since~$\cA$ is a cone, this implies that $\cA=\cX$ which is not possible because~$\cA$ is an acceptance set. The result now follows directly from Theorem~\ref{conic finiteness}.
\end{proof}

\begin{remark}
Part~(i) in Example~\ref{non reachability example} shows that if $\rho_{\cA,S}$ does not assume the value~$-\infty$, then $\rho_{\cA,S}$ need not be finitely valued. Note that the acceptance set in~(ii) and~(iii) in Example~\ref{non reachability example} are conic but not convex, showing that convexity is essential in Corollary~\ref{coherent finiteness}.
\end{remark}


\subsection{$\VaR$- and $\TVaR$-based capital requirements}

As direct applications of Theorem~\ref{conic finiteness} and Corollary~\ref{coherent finiteness}, we provide a characterization of when capital requirements based on $\VaR$- and $\TVaR$-acceptability are finite. Recall that the corresponding acceptance sets~$\cA_\alpha$ and~$\cA^\alpha$ have been introduced in Example~\ref{var tvar example}.

\medskip

\begin{proposition}[$\VaR$-acceptance]
\label{ruin and eligibility}
Let $S=(S_0,S_T)$ be a traded asset and $\alpha\in(0,1)$.
\begin{enumerate}[(i)]
	\item The following statements are equivalent:
\begin{enumerate}[(a)]
  \item $\rho_{\cA_\alpha,S}(X)<\infty$ for every $X\in\cX$;
  \item $\VaR_\alpha(S_T)<0$;
  \item $\probp[S_T<\lambda]\leq\alpha$ for some $\lambda>0$.
\end{enumerate}
In the case that $\Omega$ is finite, the above statements are equivalent to $\probp[S_T=0]\leq\alpha$.

	\item The following statements are equivalent:
\begin{enumerate}[(a)]
  \item $\rho_{\cA_\alpha,S}(X)>-\infty$ for every $X\in\cX$;
  \item $\VaR_\alpha(-S_T)>0$;
  \item $\probp[S_T>0]>\alpha$ ($\probp[S_T=0]<1-\alpha$).
\end{enumerate}
\end{enumerate}
\end{proposition}
\begin{proof}
Note that $\VaR_\alpha(S_T)<0$ is equivalent to $S_T\in\Interior(\cA_\alpha)$ and $\VaR_\alpha(-S_T)>0$ is equivalent to $-S_T\notin\overline{\cA_\alpha}$. Hence the equivalence of~(a) and~(b) in parts~(i) and~(ii) follows from Theorem~\ref{conic finiteness} since~$\cA_\alpha$ is a conic acceptance set. The equivalence of~(b) and~(c) in parts~(i) and~(ii) follows directly from the definition of Value-at-Risk. Clearly, if~$\Omega$ is finite then~\textit{(c)} is equivalent to $\probp[S_T=0]\leq\alpha$, concluding the proof.
\end{proof}

\begin{remark}
\label{remark var finiteness}
Condition~(c) in parts~(i) and~(ii) of Proposition~\ref{ruin and eligibility} are quite intuitive. We interpret them when $S=(S_0,S_T)$ represents a defaultable bond with price $S_0\in(0,1)$, face value~$1_\Omega$ and payoff~$S_T$ with $0\leq S_T \leq 1_\Omega$.

\smallskip

Condition~(c) in part~(i) implies that $\probp[S_T=0]\le \alpha$. Note that the set $\{S_T=0\}$ is precisely the set of future states on which the bond defaults fully. Hence, raising capital and investing it in the bond has no impact whatsoever on the capital position in those states. Therefore, if~$\rho_{\cA_\alpha,S}$ is not to assume the value~$\infty$, the probability of zero recovery cannot exceed~$\alpha$. Similarly, the set $\{S_T>0\}$ is precisely the set of future states where the bond pays a strictly positive amount. Raising and investing capital in the bond does have an impact on the capital position in those states. Therefore, if~$\rho_{\cA_\alpha,S}$ is not to assume the value~$-\infty$, the probability of this set needs to be greater than~$\alpha$.
\end{remark}

\medskip

As a corollary, we provide a sharper characterization of when the risk measure~$\rho_{\cA_\alpha,S}$ based on $\VaR$-acceptability is finite in case $\alpha\in\left(0,\frac{1}{2}\right)$. Note that this situation is most relevant for applications.

\begin{corollary}
\label{corollary finiteness var alpha less 0.5}
Let $S=(S_0,S_T)$ be a traded asset and take $\alpha\in\left(0,\frac{1}{2}\right)$. Then are equivalent:
\begin{enumerate}[(a)]
	\item $\rho_{\cA_\alpha,S}$ is finitely valued;
	\item $\VaR_\alpha(S_T)<0$;
  \item $\probp[S_T< \lambda]\leq\alpha$ for some $\lambda>0$.
\end{enumerate}
In the case that $\Omega$ is finite, the above statements are equivalent to $\probp[S_T=0]\leq\alpha$.
\end{corollary}
\begin{proof}
By Proposition~\ref{ruin and eligibility} we know that~\textit{(a)} is a sufficient condition for~\textit{(b)} which is equivalent to~\textit{(c)}. Assume~\textit{(c)} holds. Then, part~(i) of Proposition~\ref{ruin and eligibility} implies that~$\rho_{\cA,S}$ cannot assume the value~$\infty$. Moreover, we have $\probp[S_T>0]\geq\probp[S_T\geq\lambda]\geq 1-\alpha>\alpha$. Hence, by part~(ii) of Proposition~\ref{ruin and eligibility}, the risk measure~$\rho_{\cA,S}$ cannot assume the value~$-\infty$. This proves that~\textit{(c)} implies~\textit{(a)}, concluding the proof.
\end{proof}

\medskip

\begin{proposition}[$\TVaR$-acceptance]
\label{TVaR-acceptance}
Let $S=(S_0,S_T)$ be a traded asset and $\alpha\in(0,1)$. The following statements are equivalent:
\begin{enumerate}[(a)]
	\item $\rho_{\cA^\alpha,S}$ is finitely valued;
	\item $\TVaR_\alpha(S_T)<0$;
	\item $\probp[S_T=0]<\alpha$.
\end{enumerate}
\end{proposition}
\begin{proof}
Note that $\TVaR_\alpha(S_T)<0$ is equivalent to $S_T\in\Interior(\cA^\alpha)$. The equivalence between~\textit{(a)} and~\textit{(b)} thus follows immediately from Corollary~\ref{coherent finiteness} since~$\cA^\alpha$ is a coherent acceptance set.

\smallskip

Now assume~\textit{(b)} holds, so that $\int^{\alpha}_{0}\VaR_\beta(S_T)d\beta<0$. Since $S_T\geq0$ we have $\VaR_\beta(S_T)\leq0$ for every $\beta\in(0,\alpha)$, and therefore there must exist $\gamma\in(0,\alpha)$ such that $\VaR_\gamma(S_T)<0$. Hence, we find $\lambda>0$ for which $\probp[S_T=0]\leq\probp[S_T<\lambda]\leq\gamma<\alpha$, proving~\textit{(c)}.

\smallskip

Conversely, assume~\textit{(c)} is fulfilled. Then we can find $\lambda>0$ such that $\probp[S_T<\lambda]<\alpha$. Setting $\gamma:=\probp[S_T<\lambda]$ we have $\VaR_\gamma(S_T)\leq-\lambda<0$. This implies that $\VaR_\beta(S_T)<0$ for every $\beta\in(\gamma,\alpha)$, which immediately yields $\TVaR_\alpha(S_T)<0$, concluding the proof.
\end{proof}

\smallskip

\begin{remark}
\label{remark on tvar finiteness}
As for Proposition~\ref{ruin and eligibility}, the interpretation of Proposition~\ref{TVaR-acceptance} is straightforward when~$S$ represents a defaultable bond. In this case, the corresponding $\TVaR$-based capital requirements are finitely valued if and only if the probability of zero recovery is strictly smaller than~$\alpha$.
\end{remark}


\section{Continuity properties of risk measures}
\label{continuity section}

In this section we investigate continuity properties of risk measures. The continuity of capital requirements is important because it gives an indication of the ``stability'' of the required capital figure with respect to small perturbations of the capital position, i.e. of the reliability of required capital if we only know the capital position approximately.

\medskip

We know from Proposition~\ref{finiteness baz} that, whenever the payoff of the eligible asset~$S$ is bounded away from zero, the risk measure~$\rho_{\cA,S}$ is finitely valued and Lipschitz continuous for every acceptance set~$\cA$. The next example shows that a capital requirement may be finitely valued and continuous even if the payoff of the eligible asset is not bounded away from zero.

\smallskip

\begin{example}
Let~$\probp$ be a probability measure on $(\Omega,\cF)$ and $\alpha\in\R$, and consider the acceptance set
\begin{equation}
\cA:=\{X\in\cX \,; \ \E[X]\geq\alpha\}\,.
\end{equation}
Then~$\rho_{\cA,S}$ is finitely valued and continuous for any traded asset $S=(S_0,S_T)$ with payoff $S_T$ such that $\probp[S_T>0]>0$. Indeed, we have $\rho_{\cA,S}(X)=S_0\frac{\alpha-\E[X]}{\E[S_T]}$ for all~$X\in\cX$.
\end{example}

\medskip

Using standard arguments from convex analysis it is easy to see that finite capital requirements based on convex acceptance sets are automatically Lipschitz continuous.

\begin{theorem}
\label{convexity and continuity}
Let $\cA\subset\cX$ be a convex acceptance set and $S=(S_0,S_T)$ a traded asset. If~$\rho_{\cA,S}$ is finitely valued, then it is locally Lipschitz continuous. Moreover, if~$\cA$ is coherent and~$\rho_{\cA,S}$ is finitely valued, then~$\rho_{\cA,S}$ is globally Lipschitz continuous.
\end{theorem}
\begin{proof}
By Lemma~\ref{properties acceptance sets} the acceptance set~$\cA$ has nonempty interior. Since~$\rho_{\cA,S}$ is a convex map that is bounded from above by~$0$ on~$\cA$, we conclude that~$\rho_{\cA,S}$ is locally Lipschitz continuous by applying Theorem~5.43 and Theorem~5.44 in Aliprantis and Border~(2006).

\smallskip

Assume now that~$\cA$ is coherent. Since~$\rho_{\cA,S}$ is locally Lipschitz continuous and $\rho_{\cA,S}(0)=0$, we find a constant $L>0$ and a neighborhood of~$0$ such that $\left|\rho_{\cA,S}(X)\right|\leq L\left\|X\right\|$ for any~$X$ in that neighborhood. Since~$\rho_{\cA,S}$ is positively homogeneous, the same inequality holds for any $X\in\cX$. Now fix $X,Y\in\cX$. We have $\rho_{\cA,S}(X)\leq\rho_{\cA,S}(X-Y)+\rho_{\cA,S}(Y)$ by subadditivity, and therefore $\rho_{\cA,S}(X)-\rho_{\cA,S}(Y)\leq L\left\|X-Y\right\|$. Exchanging~$X$ and~$Y$ we conclude that~$\rho_{\cA,S}$ is globally Lipschitz continuous.
\end{proof}

\medskip

Let~$\cA$ be an arbitrary acceptance set and~$S$ a traded asset. If~$\cA$ is not convex, then the fact that~$\rho_{\cA,S}$ is finitely valued no longer automatically implies continuity. The next example shows that $\VaR$-based capital requirements may sometimes fail to be continuous, even when finite. This situation is in striking contrast with the standard cash-additive setting, where Value-at-Risk is finite and Lipschitz continuous at every position~$X\in\cX$.

\smallskip

\begin{example}
\label{example no cont of var-based}
Let $\Omega=\{\omega_1,\omega_2,\omega_3\}$ and consider a probability~$\probp$ defined on the power set of~$\Omega$. For simplicity set $p_i:=\probp[\{\omega_i\}]$ for~$i=1,2,3$. Let $\alpha\in\left(0,\frac{1}{3}\right)$, and set $p_1\leq\alpha$, $p_2\in(\alpha-p_1,\alpha]$. Since $\alpha<\frac{1}{3}$, we have $p_3>\alpha$. Then consider a defaultable bond $S=(S_0,S_T)$ with unitary price $S_0:=1$ and payoff~$S_T$ such that $S_T(\omega_1):=0$, $S_T(\omega_2)>0$, and $S_T(\omega_3)>0$. For an arbitrary $X\in\cX$, it is not difficult to show that
\begin{equation}
\rho_{\cA_\alpha,S}(X)=\left\{
\begin{array}{l l}
-\frac{X(\omega_3)}{S_T(\omega_3)} \ \ \ \ \ \ \ \ \ \ \ \ \ \ \ \ \ \ \ \,\quad\mbox{if $X(\omega_1)\geq0$}, \\
-\min\left\{\frac{X(\omega_2)}{S_T(\omega_2)},\frac{X(\omega_3)}{S_T(\omega_3)}\right\} \quad\mbox{otherwise}.
\end{array}
\right.
\end{equation}
Since $p_1=\probp[S_T=0]\leq\alpha$, the risk measure~$\rho_{\cA_\alpha,S}$ is finitely valued by Corollary~\ref{corollary finiteness var alpha less 0.5}. But it is not continuous at $X:=1_{\{\omega_3\}}$. Indeed, the sequence defined by $X_n:=X-\frac{1}{n}1_{\{\omega_1\}}$ converges to~$X$ but $\rho_{\cA_\alpha,S}(X)=-\frac{1}{S_T(\omega_3)}<0$ while $\rho_{\cA_\alpha,S}(X_n)=0$ for all~$n$.
\end{example}


\subsection{The geometry of continuity}
\label{continuity subsection}

Whether or not~$\rho_{\cA,S}$ is continuous will typically depend on the interplay between the acceptance set~$\cA$ and the eligible asset~$S$. In Theorem~\ref{pointwise semicontinuity} below we provide a general pointwise characterization of upper and lower semicontinuity, and thus also continuity.

\smallskip

Recall that a map $\rho:\cX\to\overline{\R}$ is called \textit{lower semicontinuous} at $X\in\cX$ if for every $\e>0$ there exists a neighborhood~$\cU$ of~$X$ such that $\rho(Y)\geq\rho(X)-\e$ for all $Y\in\cU$. We say that~$\rho$ is (globally) lower semicontinuous if it is lower semicontinuous at each point $X\in\cX$. Note that~$\rho$ is lower semicontinuous if and only if the set $\{X\in\cX \,; \ \rho(X)\leq\lambda\}$ is closed for every $\lambda\in\R$. The map~$\rho$ is \textit{upper semicontinuous} at $X\in\cX$ if~$-\rho$ is lower semicontinuous at~$X$ and (globally) upper semicontinuous if~$-\rho$ is lower semicontinuous. Note that~$\rho$ is continuous at $X\in\cX$ if and only it is both lower and upper semicontinuous at~$X$.

\smallskip

Moreover, we write $X_n\downarrow X$ whenever the sequence $(X_n)$ converges to $X$ and is decreasing, i.e. $X_n\geq X_{n+1}$ for all $n$. Similarly, we write $X_n\uparrow X$ if the sequence $(X_n)$ converges to $X$ and is increasing, i.e. $X_n\leq X_{n+1}$ for every $n$. Finally,~$\rho$ is called \textit{norm-continuous from above}, respectively \textit{from below}, at a point~$X$ if $\rho(X_n)\to\rho(X)$ whenever $X_n\downarrow X$, respectively $X_n\uparrow X$. Note that this continuity notion is different from the pointwise notion introduced in Lemma~4.21 in F\"{o}llmer and Schied~(2011).

\medskip

\begin{theorem}[Pointwise semicontinuity]
\label{pointwise semicontinuity}
Let $\cA\subset\cX$ be an acceptance set, $S=(S_0,S_T)$ a traded asset, and take $X\in\cX$.
\begin{enumerate}[(i)]
\item The following statements are equivalent:
\begin{enumerate}[(a)]
	\item $\rho_{\cA,S}$ is lower semicontinuous at~$X$;
	\item $X+\frac{m}{S_0}S_T\notin\overline{\cA}$ for any $m<\rho_{\cA,S}(X)$;
	\item $\rho_{\overline{\cA},S}(X)=\rho_{\cA,S}(X)$;
	\item $\rho_{\cA,S}$ is norm-continuous from above at $X$;
	\item $\rho_{\cA,S}\left(X+\frac{1}{n}1_\Omega\right)\to\rho_{\cA,S}(X)$.
\end{enumerate}
\item The following statements are equivalent:
\begin{enumerate}[(a)]
  \item $\rho_{\cA,S}$ is upper semicontinuous at~$X$;
	\item $X+\frac{m}{S_0}S_T\in\Interior(\cA)$ for any $m>\rho_{\cA,S}(X)$;
	\item $\rho_{\Interior(\cA),S}(X)=\rho_{\cA,S}(X)$;
	\item $\rho_{\cA,S}$ is norm-continuous from below at $X$;
	\item $\rho_{\cA,S}\left(X-\frac{1}{n}1_\Omega\right)\to\rho_{\cA,S}(X)$.
\end{enumerate}
\end{enumerate}
\end{theorem}
\begin{proof}
We only prove part~{\em (i)} on lower semicontinuity. The proof of part~{\em (ii)} is similar.

\smallskip

To prove that~{\em (a)} implies~{\em (b)}, note that by $S$-additivity~$\rho_{\cA,S}$ is lower semicontinuous at~$X$ if and only if for any $m<\rho_{\cA,S}(X)$ there exists a neighborhood~$\cU$ of~$X$ such that $\rho_{\cA,S}(Y+\frac{m}{S_0}S_T)>0$ for all $Y\in\cU$. Then $Y+\frac{m}{S_0}S_T\notin\cA$ for every such~$Y$, implying $X+\frac{m}{S_0}S_T\notin\overline{\cA}$.

\smallskip

If~{\em (b)} holds, then $\rho_{\overline{\cA},S}(X)\geq\rho_{\cA,S}(X)$. Since the opposite inequality is always satisfied, it follows that~{\em (c)} holds.

\smallskip

To prove that~{\em (c)} implies~{\em (d)}, assume $X_n\downarrow X$. By monotonicity there exists $\rho_0\in\R$ such that $\rho_{\cA,S}(X_n)\to\rho_0\leq\rho_{\cA,S}(X)$. Since $X_n+\frac{\rho_{\cA,S}(X_n)}{S_0}S_T\in\overline{\cA}$, it follows that $X+\frac{\rho_0}{S_0}S_T\in\overline{\cA}$. Hence, $\rho_{\cA,S}(X)=\rho_{\overline{\cA},S}(X)\leq\rho_0$. As a result, $\rho_0=\rho_{\cA,S}(X)$, proving that~{\em (d)} holds.

\smallskip

Clearly~\textit{(d)} implies~\textit{(e)}. Finally assume~\textit{(e)} holds and set $X_n:=X+\frac{1}{n}1_\Omega$. Since by assumption $\rho_{\cA,S}(X_n)\to\rho_{\cA,S}(X)$, for any $\e>0$ there exists an integer $n_\e$ for which $\rho_{\cA,S}(X_{n_\e})\geq\rho_{\cA,S}(X)-\e$. Note that the set $\cU_\e:=\{Y\in\cX \,; \ Y\leq X_{n_\e}\}$ is a neighborhood of~$X$ and, by monotonicity, $\rho_{\cA,S}(Y)\geq\rho_{\cA,S}(X)-\e$ for all $Y\in\cU_\e$. Hence,~$\rho_{\cA,S}$ is lower semicontinuous at~$X$ and~{\em (a)} holds, concluding the proof.
\end{proof}

\begin{remark}
Consider an acceptance set $\cA\subset\cX$ and a traded asset $S=(S_0,S_T)$. It is clear from Theorem~\ref{pointwise semicontinuity} how to characterize pointwise continuity of~$\rho_{\cA,S}$. We just highlight that~$\rho_{\cA,S}$ is continuous at a point $X\in\cX$ where~$\rho_{\cA,S}$ is finite if and only if
\begin{enumerate}[(i)]
  \item $X+\frac{m}{S_0}S_T\notin\overline{\cA}$ for $m<\rho_{\cA,S}(X)$;
  \item $X+\frac{m}{S_0}S_T\in\Interior(\cA)$ for $m>\rho_{\cA,S}(X)$.
\end{enumerate}
That is, if and only if the line $X+\frac{m}{S_0}S_T$ comes from outside~$\overline{\cA}$ for $m<\rho_{\cA,S}(X)$, perforates the boundary~$\partial\cA$ at $m=\rho_{\cA,S}(X)$, and immediately enters~$\Interior(\cA)$ for $m>\rho_{\cA,S}(X)$.
\end{remark}

\medskip

Next we provide a characterization of the global lower and upper semicontinuity of~$\rho_{\cA,S}$.

\begin{proposition}[Global semicontinuity]
\label{semicontinuity}
Let $\cA\subset\cX$ be an acceptance set and~$S=(S_0,S_T)$ a traded asset.
\begin{enumerate}[(i)]
\item The following statements are equivalent:
\begin{enumerate}[(a)]
  \item $\rho_{\cA,S}$ is (globally) lower semicontinuous;
  \item $\{X\in\cX \,; \ \rho_{\cA,S}(X)\le 0 \}$ is closed;
  \item $\overline{\cA}=\{X\in\cX \,; \ \rho_{\cA,S}(X)\le 0\}$\,.
\end{enumerate}
\item The following statements are equivalent:
\begin{enumerate}[(a)]
  \item $\rho_{\cA,S}$ is (globally) upper semicontinuous;
  \item $\{X\in\cX \,; \ \rho_{\cA,S}(X)<0\}$ is open;
  \item $\Interior(\cA)=\{X\in\cX \,; \ \rho_{\cA,S}(X)<0\}$\,.
\end{enumerate}
\end{enumerate}
\end{proposition}
\begin{proof}
As in the previous proposition we only prove part~(i). Clearly,~{\em (a)} implies~{\em (b)}. Moreover, by Lemma~\ref{inclusions A and rho} we have $\cA\subset\{X\in\cX \,; \ \rho_{\cA,S}(X)\le 0 \}\subset\overline{\cA}$. Therefore,~{\em (b)} implies~{\em (c)}. Finally, if~{\em (c)} holds then $\{X\in\cX \,; \ \rho_{\cA,S}(X)\le \lambda \}=\lambda S_T+\overline{\cA}$ for every $\lambda\in\R$, hence~{\em (a)} follows.
\end{proof}

\begin{remark}
\label{remark on semicont}
If~$\cA$ is a closed acceptance set, then~$\rho_{\cA,S}$ is lower semicontinuous for any choice of~$S$. If~$\cA$ is open, then~$\rho_{\cA,S}$ is upper semicontinuous for any choice of~$S$. Both properties follow easily from the preceding result and Lemma~\ref{inclusions A and rho}.
\end{remark}


\subsection{$\VaR$- and $\TVaR$-based capital requirements}

In this section we characterize the continuity of capital requirements based on Value-at-Risk and Tail-Value-at-Risk. It turns out that while for $\TVaR$-acceptability every finite capital requirement is also continuous, this is not the case for $\VaR$-acceptability. We will also show that, when the underlying probability space is nonatomic, $\VaR$-based capital requirements are continuous only when the payoff of the eligible asset is bounded away from zero. In contrast, risk measures based on $\TVaR$-acceptability may be continuous even if the eligible asset is a defaultable bond admitting zero recovery, provided the probability of zero recovery is not too high. As a consequence, when determining capital requirements with respect to a defaultable eligible asset, acceptability based on $\TVaR$ appears to be preferable, from an operational perspective related to the stability of the corresponding risk measure, to $\VaR$-acceptability.

\smallskip

The closed acceptance sets~$\cA_\alpha$ and~$\cA^\alpha$ based on Value-at-Risk and Tail Value-at-Risk have been introduced in Example~\ref{var tvar example}. Recall also that $\Interior(\cA_\alpha)=\{X\in\cX \,; \ \VaR_\alpha(X)<0\}$ and, similarly, $\Interior(\cA^\alpha)=\{X\in\cX \,; \ \TVaR_\alpha(X)<0\}$. These facts will be used without further reference.

\medskip

Since $\cA_\alpha$ is closed, Remark~\ref{remark on semicont} immediately implies the following preliminary result.

\begin{lemma}
\label{lsc of var-based risk measures}
Let $\alpha\in(0,1)$ and $S=(S_0,S_T)$ a traded asset. Then~$\rho_{\cA_\alpha,S}$ is lower semicontinuous.
\end{lemma}

\medskip

A financial position that is acceptable with respect to~$\cA_\alpha$ remains acceptable if we change it on a set of zero probability. The following proposition can be then regarded as a straightforward refinement of Proposition~\ref{finiteness baz}. Note that we are working on the space of bounded measurable functions on $(\Omega,\cF)$ and are not identifying functions which are $\probp$-almost surely identical. We will say that $X\in\cX$ is \textit{essentially bounded away from zero} if $\probp[X\geq\e]=1$ for some $\e>0$.

\begin{proposition}
\label{order unit var}
Let $S=(S_0,S_T)$ be a traded asset whose payoff~$S_T$ is essentially bounded away from zero and let $\alpha\in(0,1)$. Then~$\rho_{\cA_\alpha,S}$ is finitely valued and Lipschitz continuous.
\end{proposition}
\begin{proof}
Since~$S_T$ is essentially bounded away from zero, there exist $A\in\cF$ with $\probp[A]=1$ and $\e>0$ such that $S_T\geq\e1_A$. Set $\widetilde{S}_0:=S_0$ and $\widetilde{S}_T:=S_T+\e1_{A^c}$. Then $\widetilde{S}_T\geq\e1_\Omega$, that is $\widetilde{S}_T$ is bounded away from zero. Since $\rho_{\cA_\alpha,\widetilde{S}}=\rho_{\cA_\alpha,S}$ the statement follows immediately from Proposition~\ref{finiteness baz}.
\end{proof}

\medskip

We next characterize when $\VaR$-based capital requirements with respect to a general eligible asset are continuous. We focus on the two cases that are most interesting in applications, namely that of a finite probability space and that of a nonatomic probability space. In case~$\Omega$ is finite, we first provide the following characterization of pointwise continuity.

\begin{proposition}
\label{pointwise cont var-based risk measures}
Let $(\Omega,\cF,\probp)$ be a finite probability space and~$\alpha\in(0,1)$. Consider a traded asset~$S=(S_0,S_T)$, and a position $X\in\cX$ with $\rho_{\cA_\alpha,S}(X)$ is finite. Define $\widetilde{X}:=X+\frac{\rho_{\cA_\alpha,S}(X)}{S_0}S_T$. The following conditions are equivalent:
\begin{enumerate}[(a)]
	\item $\rho_{\cA_\alpha,S}$ is continuous at~$X$;
	\item $\probp[\widetilde{X}<0]+\probp[\{X=0\}\cap\{S_T=0\}]\leq\alpha$.
\end{enumerate}
\end{proposition}
\begin{proof}
Since $\Omega$ is a finite set, it is not difficult to show that there exists $\e_0>0$ such that for every $0<\e<\e_0$ we eventually have
\begin{equation}
\label{var cont 1}
\{\widetilde{X}<0\}\cup(\{X=0\}\cap\{S_T=0\})=\left\{\widetilde{X}-\frac{1}{n}1_\Omega+\frac{\e}{S_0}S_T<0\right\}\,.
\end{equation}
As a result, we see from~\eqref{var cont 1} that the assertion~\textit{(b)} holds if and only if for every $0<\e<\e_0$ we eventually have $\rho_{\cA,S}(\widetilde{X}-\frac{1}{n}1_\Omega)\leq\e$, or equivalently $\rho_{\cA,S}(\widetilde{X}-\frac{1}{n}1_\Omega)\to 0$. By Theorem~\ref{pointwise semicontinuity}, this is equivalent to $\rho_{\cA,S}$ being upper semicontinuous at $\widetilde{X}$, hence at $X$, by $S$-additivity. Finally, since $\rho_{\cA,S}$ is globally lower semicontinuous by Lemma~\ref{lsc of var-based risk measures}, upper semicontinuity at $X$ is in fact equivalent to~\textit{(a)}, concluding the proof.
\end{proof}

\begin{remark}
\begin{enumerate}[(i)]
	\item The fact that condition~\textit{(b)} in the preceding proposition is equivalent to pointwise continuity at~$X$ becomes clear if we consider~\eqref{var cont 1}. In particular, if~\textit{(b)} is not satisfied, the capital requirement of~$X-\frac{1}{n}1_\Omega$ does not converge to the capital requirement of~$X$.
	\item By definition of~$\widetilde{X}$, we always have $\probp[\widetilde{X}<0]\leq\alpha$. For~$\rho_{\cA_\alpha,S}$ to be continuous at~$X$, condition~\textit{(b)} tells that the set $\{X=0\}\cap\{S_T=0\}$ should not add too much probability to $\{\widetilde{X}<0\}$. In particular, if $\probp[\{X=0\}\cap\{S_T=0\}]=0$ the risk measure~$\rho_{\cA_\alpha,S}$ is continuous at~$X$.
\end{enumerate}
\end{remark}

\smallskip

\begin{remark}
In the light of the previous result, we can better understand why the risk measure in Example~\ref{example no cont of var-based} is not continuous at $X:=1_{\{\omega_3\}}$. There, $\probp[X<0]=p_2\leq\alpha$, but $\{X=0\}\cap\{S_T=0\}=\{\omega_1\}$ has probability $p_1$, which is too high since $p_1+p_2>\alpha$.
\end{remark}

\medskip

Next we provide a global continuity result for $\VaR$-based capital requirements.

\begin{proposition}
\label{continuity var based}
Let $(\Omega,\cF,\probp)$ be a finite probability space and~$\alpha\in(0,1)$. Let~$S=(S_0,S_T)$ be a traded asset, and assume $\rho_{\cA_\alpha,S}$ is finitely valued. The following conditions are equivalent:
\begin{enumerate}[(a)]
	\item $\rho_{\cA_\alpha,S}$ is continuous on~$\cX$;
	\item $\probp[A]\leq\alpha-\probp[S_T=0]$ for every measurable set $A\subset\{S_T>0\}$ with $\probp[A]\leq\alpha$.
\end{enumerate}	
\end{proposition}
\begin{proof}
For $X\in\cX$ we use the notation~$\widetilde{X}$ introduced in Proposition~\ref{pointwise cont var-based risk measures}. Note that $\widetilde{X}\in\cA_\alpha$ implies $\probp[\widetilde{X}<0]\leq\alpha$.

\smallskip

To prove that~\textit{(a)} implies~\textit{(b)}, take a measurable set $A\subset\{S_T>0\}$ with $\probp[A]\leq\alpha$, and set $X:=-1_A$. Note that $\probp[X<0]=\probp[A]\leq\alpha$, implying $\rho_{\cA_\alpha,S}(X)\leq0$. Moreover, since $X\leq0$, we actually have $\rho_{\cA_\alpha,S}(X)=0$, so that $\widetilde{X}=X$. Since $\rho_{\cA_\alpha,S}$ is continuous at~$X$ by assumption, it follows from Proposition~\ref{pointwise cont var-based risk measures} that
\begin{equation}
\label{var cont global 0}
\probp[A]+\probp[S_T=0] = \probp[\widetilde{X}<0]+\probp[\{X=0\}\cap\{S_T=0\}] \leq \alpha\,,
\end{equation}
showing that~\textit{(b)} holds.

\smallskip

To prove that~\textit{(b)} implies~\textit{(a)}, take $X\in\cX$ and define $A:=\{\widetilde{X}<0\}\cap\{S_T>0\}$. Note that
\begin{equation}
\label{var cont global 1}
\{\widetilde{X}<0\}\cup(\{X=0\}\cap\{S_T=0\})\subset A\cup\{S_T=0\}\,.
\end{equation}
Since $\probp[A]\leq\probp[\widetilde{X}<0]\leq\alpha$, we obtain by assumption that $\probp[A]+\probp[S_T=0]\leq\alpha$. Hence it follows from~\eqref{var cont global 1} that $\probp[\widetilde{X}<0]+\probp[\{X=0\}\cap\{S_T=0\}]\leq\alpha$, and therefore Proposition implies~$\rho_{\cA_\alpha,S}$ is continuous at~$X$, concluding the proof.
\end{proof}

\begin{remark}
\begin{enumerate}
	\item Condition~\textit{(b)} in the preceding proposition provides a simple sufficient condition for the continuity of risk measures based on $\VaR$-acceptability, and can be seen as the global version of condition~\textit{(b)} in Proposition~\ref{pointwise cont var-based risk measures}, as shown by~\eqref{var cont global 0} and~\eqref{var cont global 1}.
	\item Note that condition~\textit{(b)} is trivially satisfied whenever $\probp[S_T=0]=0$, which in the finite setting is equivalent to~$S_T$ being essentially bounded away from zero.
\end{enumerate}
\end{remark}

\medskip

Perhaps surprisingly, the range of eligible assets allowing for continuous $\VaR$-based risk measures is much narrower when the underlying probability space is nonatomic. In this case, we show that capital requirements based on $\VaR$-acceptability are globally continuous if and only if the payoff of the eligible asset is essentially bounded away from zero. As a result, when the eligible asset is a defaultable bond admitting zero recovery, or more generally a recovery which is not bounded away from zero with full probability, we never have continuity. Recall that a probability space $(\Omega,\cF,\probp)$ is nonatomic if it has no atoms, i.e. if it has no sets $A\in\cF$ with $\probp[A]>0$ such that either $\probp[B]=0$ or $\probp[B]=\probp[A]$ whenever~$B$ is a measurable subset of~$A$. For more details, see Appendix~A.3 in F\"{o}llmer and Schied~(2011).

\begin{proposition}
\label{var and continuity}
Let $(\Omega,\cF,\probp)$ be a nonatomic probability space and~$\alpha\in(0,1)$. Consider a traded asset~$S=(S_0,S_T)$ and assume~$\rho_{\cA_\alpha,S}$ is finitely valued. The following two statements are equivalent:
\begin{enumerate}[(a)]
	\item $\rho_{\cA_\alpha,S}$ is continuous on~$\cX$;
	\item $S_T$ is essentially bounded away from zero.
\end{enumerate}
\end{proposition}
\begin{proof}
By Proposition~\ref{order unit var} we only need to prove that~\textit{(a)} implies~\textit{(b)}. Assume to the contrary that the payoff~$S_T$ is not essentially bounded away from zero, so that $\probp[S_T<\lambda]>0$ for every $\lambda>0$. Clearly there exists $\lambda_0>0$ such that $0<\probp[S_T<\lambda_0]<1-\alpha$. Since $(\Omega,\cF,\probp)$ is nonatomic, we can find a measurable set $A\subset\{S_T\geq\lambda_0\}$ with $\probp[A]=\alpha$. Then set $X:=-\left\|S_T\right\|1_A$, and note that $\VaR_\alpha(X)=0$, hence $\rho_{\cA_\alpha,S}(X)\leq0$. Moreover, for any $0<\lambda\leq\lambda_0$ we clearly have $\{X+S_T<\lambda\}\cap A=A$ and $\{X+S_T<\lambda\}\cap A^c=\{S_T<\lambda\}$. It follows that for $0<\lambda\leq\lambda_0$
\begin{equation}
\probp[X+S_T<\lambda]=\probp[A]+\probp[S_T<\lambda]>\alpha\,.
\end{equation}
Hence, $\VaR_\alpha(X+S_T)\ge 0$ so that $X+S_T\not\in\Interior(\cA_\alpha)$. As a result, Theorem~\ref{pointwise semicontinuity} implies that $\rho_{\cA_\alpha,S}$ is not upper semicontinuous and, thus, not continuous at~$X$. It follows that if~$\rho_{\cA_\alpha,S}$ is continuous, then~$S_T$ must be essentially bounded away from zero.
\end{proof}

\medskip

Next we consider $\TVaR$-based capital requirements. Since the corresponding acceptance set~$\cA^\alpha$ is closed, Remark~\ref{remark on semicont} immediately implies the following result.

\begin{lemma}
Let $\alpha\in(0,1)$ and $S=(S_0,S_T)$ a traded asset. Then~$\rho_{\cA^\alpha,S}$ is lower semicontinuous.
\end{lemma}

\medskip

By applying Theorem~\ref{convexity and continuity} we obtain a characterization of continuity for $\TVaR$-based capital requirements. Comparing this result with the preceding continuity results for $\VaR$-based risk measures, especially Proposition~\ref{var and continuity}, we conclude that, when determining capital requirements with respect to a defaultable eligible asset, acceptability based on $\TVaR$ appears to be preferable to $\VaR$-acceptability in terms of the corresponding risk measure stability.

\begin{proposition}
\label{continuity of TVaR}
Let $\alpha\in(0,1)$ and $S=(S_0,S_T)$ a traded asset. Then are equivalent:
\begin{enumerate}[(a)]
	\item $\rho_{\cA^\alpha,S}$ is (globally) Lipschitz continuous;
	\item $\TVaR_\alpha(S_T)<0$;
	\item $\probp[S_T=0]<\alpha$.
\end{enumerate}
\end{proposition}

\begin{remark}
Note that capital requirements based on $\TVaR$-acceptability may be continuous even if the payoff of the eligible asset is not essentially bounded away from zero, and even in the extreme situation where the payoff is actually zero in some future scenario, provided the probability of full default is strictly controlled by $\alpha$. In the case that the underlying probability space is nonatomic, this is in apparent contrast to the $\VaR$-case, for which continuity is guaranteed only when the payoff of the eligible asset is essentially bounded away from zero.
\end{remark}


\section{Capital efficiency and optimal eligible assets}
\label{choice of eligible}

In this last section we address the following question: given an acceptance set, can we optimally choose the eligible asset in such a way that the corresponding required capital is lower than the required capital compatible with the same acceptance set but defined by a different eligible asset?


\subsection{Equality of risk measures}
\label{equality of risk measures}

We start by characterizing when two lower semicontinuous risk measures coincide. Recall that capital requirements based on $\VaR$- and $\TVaR$-acceptability are always lower semicontinuous.

\smallskip

For later use, we introduce the following notation. If~$S=(S_0,S_T)$ and~$R=(R_0,R_T)$ are two traded assets with the same initial price $P:=S_0=R_0$, we define for $m\in\R$
\begin{equation}
\cM_m(S,R):=\left\{\frac{m}{P}S_T+\lambda (R_T-S_T) \,; \ \lambda\in\R\right\}\,.
\end{equation}
Every element in~$\cM_m(S,R)$ represents the payoff of a possible portfolio we can set up by trading in~$S$ and~$R$ with initial capital~$m$.

\begin{proposition}
\label{equality for general measures}
Consider two acceptance sets $\cA,\cB\subset\cX$ and two traded assets $S=(S_0,S_T)$ and $R=(R_0,R_T)$ with the same initial price $S_0=R_0$. Assume that~$\rho_{\cA,S}$ and~$\rho_{\cB,R}$ are both lower semicontinuous. Then $\rho_{\cA,S}=\rho_{\cB,R}$ if and only if
\begin{equation}
\overline{\cA}=\overline{\cB} \ \ \ \mbox{and} \ \ \ \overline{\cA}=\overline{\cA}+\cM_0(S,R)\,.
\end{equation}
\end{proposition}
\begin{proof}
By Proposition~\ref{semicontinuity} we may assume without loss of generality that~$\cA$ and~$\cB$ are both closed.

\smallskip

To prove necessity assume that $\rho_{\cA,S}=\rho_{\cB,R}$. Note that $\cA=\cB$ as a consequence of Proposition~\ref{semicontinuity}. It remains to show that if $X\in\cA$ and $\lambda\in\R$, then $X+\lambda(S_T-R_T)\in\cA$. But this follows from Proposition~\ref{semicontinuity}, since by additivity
\begin{equation}
\rho_{\cA,S}(X+\lambda(S_T-R_T))=\rho_{\cA,S}(X-\lambda R_T)-\lambda S_0=\rho_{\cB,R}(X)+\lambda R_0-\lambda S_0=\rho_{\cA,S}(X)\leq0\,.
\end{equation}

To prove sufficiency take $X\in\cX$. For all $m\in\R$ we have $X+\frac{m}{S_0}S_T=X+\frac{m}{R_0}R_T+\frac{m}{S_0}(S_T-R_T)$. Therefore, $X+\frac{m}{R_0}R_T\in\cA$ implies $X+\frac{m}{S_0}S_T\in\cA$. It follows that $\rho_{\cA,S}(X)\leq\rho_{\cB,R}(X)$. By exchanging the roles of~$S$ and~$R$ we obtain the reverse inequality, concluding the proof.
\end{proof}

\medskip

Under additional assumptions we can obtain a sharper characterization of the equality of risk measures.

\begin{corollary}
\label{strong equality of risk measures}
Let $\cA,\cB\subset\cX$ be two acceptance sets. Consider two traded assets $S=(S_0,S_T)$ and~$R=(R_0,R_T)$ with the same price $P:=S_0=R_0$. Assume that~$\rho_{\cA,S}$ and~$\rho_{\cB,R}$ are lower semicontinuous. Furthermore, assume that $\rho_{\cA,S}(0)$ is finite and that
\begin{equation}
\label{general no leverage condition}
\cM_{\rho_{\cA,S}(0)}(S,R)\cap\overline{\cA}=\left\{\frac{\rho_{\cA,S}(0)}{P}\,S_T\right\}\,.
\end{equation}
Then $\rho_{\cA,S}=\rho_{\cB,R}$ if and only if $\overline{\cA}=\overline{\cB}$ and $S_T=R_T$.
\end{corollary}
\begin{proof}
The ``if'' part is obvious. To prove the ``only if'' part we note that from Proposition \ref{equality for general measures} we have that $\overline{\cA}=\overline{\cB}$ and $\overline{\cA}+\cM_0(S,R)=\overline{\cA}$.
Hence, for all $\lambda\in\R$
\begin{equation}
\frac{\rho_{\cA,S}(0)}{P}\,S_T+\lambda (R_T-S_T)\in\cM_{\rho_{\cA,S}(0)}\cap\overline{\cA}\,.
\end{equation}
and by our assumption we conclude that $S_T=R_T$ must hold.
\end{proof}

\begin{remark}
To interpret condition \eqref{general no leverage condition}
it is best to consider the case $\rho_{\cA,S}(0)=0$  which was called ``non-acceptability of leverage''in Artzner et al.\,(2009) and is described in Remark~1 following Proposition~1-a in that paper. Note again that the positions in $\cM_0(S,R)$ correspond to portfolios of~$S$ and~$R$ with zero initial value. These are {\em fully-leveraged} positions where either the purchase of an amount of asset~$S$ is financed by ``borrowed money'' obtained by entering a corresponding short position in the asset~$R$, or vice versa. Thus,~\eqref{general no leverage condition} is equivalent to requiring that no fully-leveraged position of assets~$S$ and~$R$ is acceptable. This can be viewed as a reasonably natural requirement since, typically, it is considered imprudent for a ``lender'' to lend money if the ``borrower'' does not have any additional capital to support the leveraged position. Hence, regulators are unlikely to accept fully-leveraged positions.
\end{remark}


\subsection{Absence of an optimal eligible asset}
\label{general optimality section}

The next theorem is the main result of this section and relies on the above characterization of the equality of lower semicontinuous risk measures. For a fixed acceptance set $\cA\subset\cX$, we cannot find two different traded assets~$S$ and~$R$ such that $\rho_{\cA,S}(X)\leq\rho_{\cA,R}(X)$ for all positions $X\in\cX$ and such that the inequality is strict for some position. In other words, there exists no {\em optimal} eligible asset~$S$. This implies that if a regulator allows financial institutions to make a position acceptable by raising capital and, irrespective of their individual balance sheets, investing this capital amount in the \textit{same} eligible asset --- for instance the risk-free security if it exists --- then some institutions may be forced to reach acceptability at a higher cost than would have been possible by choosing an alternative eligible asset.

\begin{theorem}
\label{general optimality theorem}
Let $\cA\subset\cX$ be an acceptance set and~$S=(S_0,S_T)$ and~$R=(R_0,R_T)$ two traded assets with the same initial price $S_0=R_0$. Assume that~$\rho_{\cA,S}$ and~$\rho_{\cA,R}$ are lower semicontinuous. If $\rho_{\cA,S}(X)\leq\rho_{\cA,R}(X)$ for every position $X\in\cX$, then $\rho_{\cA,S}=\rho_{\cA,R}$.
\end{theorem}
\begin{proof}
By Proposition~\ref{equality for general measures}, it suffices to show that $\overline{\cA}=\overline{\cA}+\{\lambda(S_T-R_T) \,; \ \lambda\in\R\}$ holds. Take $\lambda\in\R$ and $X\in\overline{\cA}$, and recall that $\rho_{\cA,R}(X)\le 0$ by Proposition~\ref{semicontinuity}. Since $\rho_{\cA,R}$ dominates~$\rho_{\cA,S}$, additivity implies $\rho_{\cA,S}(X+\lambda (S_T-R_T))\leq0$. Hence, again by Proposition~\ref{semicontinuity}, $X+\lambda(R_T-S_T)$ turns out to belong to $\overline{\cA}$, completing the proof.
\end{proof}

\medskip

If we require that fully-leveraged positions in~$S$ and~$R$ are not acceptable a stronger statement can be made. In this case the only possibility to have $\rho_{\cA,S}(X)\leq\rho_{\cA,R}(X)$ for all $X\in\cX$ is that~$S$ and~$R$ have even the same payoff. The result follows immediately from Theorem~\ref{general optimality theorem} and Corollary~\ref{strong equality of risk measures}.

\begin{corollary}
\label{optimality corollary}
Let $\cA\subset\cX$ be an acceptance set and~$S=(S_0,S_T)$ and~$R=(R_0,R_T)$ two traded assets with the same price $P:=S_0=R_0$. Assume that~$\rho_{\cA,S}$ and~$\rho_{\cB,R}$ are lower semicontinuous. Furthermore, assume that $\rho_{\cA,S}(0)$ is finitely valued and that
\begin{equation}
\cM_{\rho_{\cA,S}(0)}\cap\overline{\cA}=\left\{\frac{\rho_{\cA,S}(0)}{P}\,S_T\right\}\,.
\end{equation}
If $\rho_{\cA,S}(X)\leq\rho_{\cA,R}(X)$ for every $X\in\cX$ then $S_T=R_T$.
\end{corollary}

\smallskip

\begin{remark}
In contrast to the corresponding result in Artzner et al.\,(2009), Theorem~\ref{general optimality theorem} and Corollary~\ref{optimality corollary} do not only apply to coherent risk measures but also to general risk measures, convex and not convex, including $\VaR$-based capital requirements.
\end{remark}

\smallskip

\begin{remark}
In Filipovi\'{c}~(2008) a different sort of optimality is addressed. That paper deals with a fixed convex risk measure~$\rho:\cX\to\R$ and shows there is no numeraire $U\in\cX$ which minimizes $\rho(X/U)$ for every $X\in\cX$. Note that, in this setting, $\rho$ is kept fixed and the numeraire, i.e. the unit of account, changes. Hence, with every choice of the numeraire the set of acceptable positions may change. By contrast, in our framework the acceptability criterion is independent of the particular choice of the unit of account and the natural question to ask is: what happens if we keep the acceptability criterion fixed and vary the eligible asset, i.e. the asset we are allowed to invest in?
\end{remark}

\section{Conclusion}
In the context of spaces of bounded measurable functions we have studied  finiteness and continuity properties of capital requirements with respect to general acceptance sets and general eligible assets. Considering general acceptance sets is necessary in order to include both Value-at-Risk and Tail-Value-at-Risk acceptability criteria which constitute the two most prominent examples found in the financial industry. Considering general eligible assets is important because, in contrast to the cash-additive setting, it allows for the more realistic situation where the eligible asset is a defaultable bond with a recovery rate which is not bounded away from zero. Our results and examples show that --- depending on the interplay between the acceptance set and the eligible asset --- general capital requirements display a wider range of behaviors in terms of finiteness and continuity than classical cash-additive risk measures. Hence, the theory for general risk measures turns out to be richer than that for cash-additive risk measures. We have also established that, regardless of the acceptability criterion, it is never possible to find an optimal eligible asset in the sense that choosing any other eligible asset would always result in higher capital requirements.


\end{document}